\documentclass[12pt,a4paper]{article}
\usepackage{url}
\usepackage[utf8]{inputenc}
\usepackage{tikz}
\usepackage{cancel}
\usepackage{color}
\usepackage[normalem]{ulem}
\usepackage[english]{babel}
\usepackage{color}
\usepackage[english]{babel}
\usepackage{color}
\usepackage{graphicx}
\usepackage{amsmath,amssymb,amsthm}
\usepackage[T1]{fontenc}
\usepackage[utf8]{inputenc}
\usepackage{authblk}
\usepackage[english]{babel}
\usepackage{amsmath,amssymb,amsthm}
\usepackage{color}
\usepackage{graphics}
\usepackage{amsmath,amssymb,amsthm}
\usepackage{lipsum}
\usepackage{graphicx}
\usepackage{amsmath,amssymb,bm}
\usepackage{url}
\newcommand*\rectangled[1]{%
      \tikz[baseline=(R.base)]\node[draw,rectangle,inner sep=0.9pt](R) {#1};\!
}
\newcommand*\rund[1]{%
      \tikz[baseline=(R.base)]\node[draw,circle,inner sep=0.5pt](R) {#1};\!
}

\newtheorem{theorem}{Theorem}
\newtheorem{proposition}[theorem]{Proposition}

\theoremstyle{definition}
\newtheorem{definition}[theorem]{Definition}
\newtheorem{remark}[theorem]{Remark}
\newtheorem{example}{Example}

\title{Considerate Ramp Secret Sharing}

\author{Olav Geil
}
\affil{Department of Mathematical Sciences\\ Aalborg University}

\begin{document}
\maketitle

\begin{abstract}
In this work we revisit the fundamental findings by Chen et al.\ in \cite{Chen} on general information transfer in linear ramp secret sharing schemes to conclude that their method 
not only gives a way to establish worst case leakage~\cite{Chen,kurihara} and best case recovery~\cite{Chen,geil2014relative}, but can also lead to additional insight on non-qualifying sets for any prescribed amount of information. We then apply this insight to schemes defined from monomial-Cartesian codes and by doing so we demonstrate that the good schemes from~\cite[Sec.\ IV]{MR3782266} have a second layer of security. Elaborating further, when given designed partial recovery numbers, in a new construction the focus is entirely on ensuring that the access structure possesses desirable second layer security, rather on what is the worst case information leakage in terms of {\textit{number}} of participants. The particular structure of largest possible sets being not able to determine any amount of information suggests that we coin the concept of considerate ramp secret sharing schemes of which the proposed new construction is a well-structured example. \\


\end{abstract}
\section{Introduction}\label{secintro}

A secret sharing scheme is a cryptographic method for sharing a secret among a group of participants in such a way that only certain subsets can gain full information on it. The data provided by the dealer to a participant is referred to as a share and sets of participants which are able to recover the secret by pooling their data are said to be authorized or qualified. The concept was introduced by Shamir in his seminal paper~\cite{shamir1979share} where he describes what is now known as Shamir's secret sharing scheme. His scheme which is based on univariate polynomials is perfect meaning that subsets that are not authorized possess no information on the secret. Furthermore,  Shamir's scheme is a thresshold scheme, meaning it is the size of a set of shares which indicates if it is authorized or not. 

The concept has been generalized in a variety of directions leading to much more complex constructions than the one by Shamir, for some examples see~\cite{ito1989secret,stinson1992explication,massey1993minimal,democratic1995,stinson1998bibliography,cramer2015secure,ccalkavur2022code,krenn2023introduction}. 
Here, we shall concentrate on linear schemes which can be characterized by having the property that the set of secrets equals ${\mathbb{F}}_q^\ell$ (i.e. consists of vectors of length $\ell$ over the finite field ${\mathbb{F}}_q$) whereas the shares provided to the $n$ participants  $\{ 1, \ldots , n\}$ each belongs to ${\mathbb{F}}_q$ in such a way that for two given secrets and corresponding sets of shares a linear combination of the sets of shares is an allowed set of shares for the same linear combination of the secrets. As we shall recall later in the paper, a linear scheme has the advantage that an authorized set can fastly recover the secret by means of simple linear algebra. Shamir's scheme is an example of a linear scheme where the size of the secret is the same as that of the shares, i.e.\ $\ell=1$. Schemes for which $\ell>1$ are said to be ramp, but sometimes the name is used for the entire set of linear schemes to emphasize that $\ell$ may not need to equal $1$. As we shall recall later in the paper linear schemes with $\ell=1$ are precisely those linear schemes which are perfect. Some of the earliest examples of (non-perfect) ramp secret sharing schemes were described by Blakley and Meadows in~\cite{Blakley} and by Yamamoto in~\cite{Yamamoto}, including the natural extension of thresshold schemes to a ramp version where the size of a set of participants indicates how much information it holds. Some more recent results on ramp schemes include \cite{kurosawa1994nonperfect,iwamoto2004general,iwamoto2006strongly,Chen,duursma2008algebraic,kurihara,geil2014relative,geil2017asymptotically,martinez2018communication,eriguchi2020linear}. Ramp secret sharing with $\ell >1$ is of particular interest in connection with storage of bulk data~\cite{Csirmaz} and in connection with secure multiparty computation~\cite{Chen}.

As is well-known, for ramp schemes the information held by a set of participants is discretized in the following way. Pooling their given shares a solution space $\Gamma = \vec{p}+V \subseteq {\mathbb{F}}_q^\ell$, where $V$ is a vectorspace, can be calculated the elements of $\Gamma$ being all possible secrets producing their given shares, and no vectors outside $\Gamma$ being in accordance with the shares. One say that the set of participants have $\ell - \dim V$ times $\log_{2}(q)$ bits of information. Linearity ensures that the information held by a given set of participants is a fixed number, i.e. is independent of the given secret. Hence, for $i=0, 1, \ldots , \ell$ we define $A_i$ to be the sets of participants holding $i$ times $\log_2(q)$ bits of information, but not $i+1$ times  $\log_2(q)$ bits~\cite[Def.\ 1]{iwamoto2006strongly}. Full information on  a given secret sharing scheme is equivalent to knowing $\{A_0, A_1, \ldots , A_\ell \}$ which we shall call the access structure~\cite{iwamoto2006strongly}. Except for very simple cases, the task of determining the entire access structure is a very difficult. Therefore, one often only considers the following key parameters of a secret sharing scheme, namely the privacy number $t$ and the recovery number $r$. Here, $t$ is the largest number such that no set of $t$ participants can recover any information on the secret and where $r$ is smallest possible such that any set of $r$ participants can recover the secret in full. A more refined description of the scheme is given by the parameters $t=t_1, \ldots , t_\ell$ and $r_1, \ldots , r_\ell=r$, respectively, related in a similar way to partial leakage and recovery, respectively~\cite{kurihara,geil2014relative}.

A fundamental description of linear ramp schemes where provided by Chen et al.\ in~\cite{Chen}. Here, it is shown that there is a one-to-one correspondence between these structures and sets of nested linear codes, \cite[Subsec.\ 4.2]{Chen}. Moreover, by combining~\cite[Th.\ 10]{Chen} with fundamental results by Forney~\cite{forney94} one obtains a description of the $r_1, \ldots r_\ell$ and $ t_1, \ldots , t_\ell$, respectively, in terms of relative generalized Hamming weights of the nested codes and their duals, respectively~\cite{Chen,Bains,kurihara,geil2014relative}. 

In the present paper we slightly reformulate the basic result in~\cite[Th.\ 10]{Chen} on information transfer and thereby obtains, what we believe is, a more direct way to establish further information on the access structure. To the best of our knowledge such general insight has not been employed in the literature although it seems fair to assume that it is known by more researchers in the area. By definition the maximal size of a set of participants not belonging to $A_i$ equals $r_i-1$ which may be significantly larger than $t_i$. We call such sets maximum non-$i$-qualifying ($i=1, \ldots, \ell$) and note that when they have a systematic structure we may have a way of avoiding leakage of $i$ times $\log_2(q)$ bits of information also if much more than $t_i$ participants are allowed to pool their shares. Employing our reformulation of~\cite[Th.\ 10]{Chen} we establish the systematic structure of maximum non-i-qualifying sets for a family of good secret sharing schemes based on monomial-Cartesian codes~\cite[Sec.\ IV]{MR3782266} giving rise to a second layer of security. Inspired by this analysis we next introduce the novel construction of monomial-Cartesian ramp schemes of type C where the focus is entirely on establishing schemes having maximum non-$i$-qualifying sets of a certain systematic structure which may be of interest in practical applications where one does not want to allow for ``systematic discrimination'' of the participants. Motivated by this construction we coin the larger novel concept of considerate ramp secret sharing which are schemes that protect against systematic discrimination in connection with recovery of smallest possible amount of non-zero information. The paper also contains, what we believe is, a simpler and more direct proof than can be found in the literature of the relationship between the numbers $t_1, \ldots , t_\ell, r_1, \ldots , r_\ell$ and the relative generalized Hamming weights of the nested codes and of their duals. In contrast to the literature we avoid the use of the relative distance length profile (RDLP) as well as the concept of mutual information.

The paper is organized as follows. In Section~\ref{secsss} we treat general linear ramp secret sharing schemes and provide a self-contained proof of the mentioned reformulation of~\cite[Th.\ 10 ]{Chen}. Then in Section~\ref{secpolsss} we treat the schemes from~\cite[Sec.\ IV]{MR3782266} establishing the systematic structure of maximum non-i-qualifying sets, and next in Section~\ref{secdem} we introduce monomial-Cartesian ramp schemes of type C and coin the concept of considerate ramp secret sharing. Section~\ref{secconcl} is the conclusion where we propose further research. Finally, Appendix~\ref{secproofof} contains a simplified proof of the role of relative generalized Hamming weights in connection with ramp secret sharing. 

\section{Linear ramp secret sharing schemes}\label{secsss}

In~\cite{Chen}[Sec.\ 4.2] Chen et al. presented what they call ``a more fruitful approach'' to linear ramp secret sharing schemes, namely the below coset construction: 
\begin{theorem}
The following description captures the entire set of linear ramp secret sharing schemes.
Consider a set of nested linear codes $C_2\subseteq C_1 \subseteq {\mathbb{F}}_q^n$ of codimension $\ell$. Let $\{\vec{b}_1, \ldots , \vec{b}_{k_2}\}$ be a basis for $C_2$ and $\{\vec{b}_1, \ldots , b_{k_2}, \vec{b}_{k_2+1}, \ldots , \vec{b}_{k_1=k_2+\ell}\}$ a basis for $C_1$. A secret $\vec{s} =(s_1, \ldots , s_\ell) \in {\mathbb{F}}_q^\ell$ is encoded to $\vec{c}=(c_1, \ldots , c_n)= a_1 \vec{b}_1+\cdots +a_{k_2}\vec{b}_{k_2}+ s_1\vec{b}_{k_2+1}+\cdots +s_\ell \vec{b}_{k_1}$ where $(a_1,\ldots , a_{k_2})$ is chosen uniformly at random from ${\mathbb{F}}_q^{k_2}$, and $c_i$, $i=1, \ldots , n$ are used as shares, $c_i$ being given to participant $i$ from the set of participants ${\mathcal{I}}=\{1, \ldots , n\}$. 
\end{theorem}
Below we recall~\cite{Chen}[Th.\ 10] which provides us with an exact measure for the uncertainty of the secret given any set of shares. This theorem uses the notion of the projection of a code $C \subseteq {\mathbb{F}}_q^n$ onto $A=\{i_1 < \cdots <i_{\# A}\} \subseteq {\mathcal{I}}=\{1, \ldots , n\}$ which is defined by ${\mathcal{P}}_A\big( (c_1, \ldots , c_n) \big)=(p_1, \ldots , p_n)$, where $p_i=c_i$ whenever $i \in A$ and $p_i=0$ otherwise, and where ${\mathcal{P}}_A(C)=\{\mathcal{P}_A(\vec{c}) \mid \vec{c} \in C\}$.

\begin{theorem}\label{theuncertain}
For a set of participants $A \subseteq {\mathcal{I}}=\{1, \ldots , n\}$ the uncertainty of the secret equals 
\begin{equation*}
\ell - \dim {\mathcal{P}}_A (C_1)+ \dim {\mathcal{P}}_A (C_2).
\end{equation*}
Here, an uncertainty of $a$ means that the set $A$ of participants holds $\ell-a$ times $\log_2(q)$ bits of information.
\end{theorem}

Given a linear code $C$ and $A\subseteq {\mathcal{I}}=\{1, \ldots , n\}$ recall the notation $C_A=\{\vec{c} \in C \mid c_i=0 {\mbox{ for all }} i \in \bar{A}\}$ where $\bar{A}={\mathcal{I}}\backslash A$. 
Then by applying Forney's first duality lemma~\cite{forney94}[Lem.\ 1], i.e. the result: 
\begin{equation}
\dim C = \dim C_{\bar{A}} + \dim {\mathcal{P}}_{A} (C), \label{eqfn}
\end{equation}
one can immediately translate
Theorem~\ref{theuncertain} into Theorem~\ref{theuncle} below which in our opinion is more operational when trying to detect which particular groups of participants can recover how much. This theorem uses the notion of the support ${\mbox{Supp}}(D)$ of a vector space $D \subseteq {\mathbb{F}}_q^n$ being the indexes for which the corresponding entry of at least one word in $D$ is non-zero. For the sake of self containment we provide a direct proof.

\begin{theorem}\label{theuncle}
Let $A=\{i_1 < \cdots < i_m\} \subseteq {\mathcal{I}}$. Assume $c_{i_1}^\prime \ldots , c_{i_m}^\prime$ is a set of realizable shares in those positions. I.e. there exists at least one word $\vec{c}=(c_1, \ldots , c_n)=
a_1 \vec{b}_1+\cdots +a_{k_2} \vec{b}_{k_2}+s_1 \vec{b}_{k_2+1}+ \cdots +s_\ell \vec{b}_{k_2=k_1+\ell}$, such that  $c_{i_1}=c^\prime_{i_1}, \ldots , c_{i_m}=c^\prime_{i_m}$. The amount of possible secrets $(s_1, \ldots , s_\ell)$ as above equals $q^s$ where
\begin{eqnarray}
s&=&\dim (C_1)_{\bar{A}}-\dim (C_2)_{\bar{A}} \nonumber \\
&=&\max \{ \dim D \mid D \subseteq C_1, D \cap C_2 =\{\vec{0}\}, {\mbox{Supp}}(D) \subseteq \bar{A} \} \label{eqmaxD}
\end{eqnarray}
\end{theorem}

\begin{proof}
Consider the generator matrix
$$G=\left[ \begin{array}{c} 
\vec{b}_1 \\
\vdots \\
\vec{b}_{k_2} \\
\vec{b}_{k_2+1} \\
\vdots \\ 
\vec{b}_{k_1=k_2+\ell}
\end{array} \right]$$
and let $G^\prime$ be the submatrix consisting of columns $i_1, \ldots , i_m$. Let $\vec{p} \in {\mathbb{F}}_q^{k_1}$ be a particular solution to $\vec{m} G^\prime = (c_{i_1}^\prime, \ldots , c_{i_m}^\prime)$. Then the set of possible vectors $\vec{m}$ such that $\vec{m}G^\prime = (c_{i_1}^\prime, \ldots , c_{i_m}^\prime)$ equals $\vec{p}+V$ where $V$ is the (left) kernel of $G^\prime$, i.e. $V$ is isomorphic to ${(C_1)}_{\bar{A}}$. However, we are only interested in the space consisting of the restriction of $\vec{m}=(\vec{a},\vec{s})$ to the last $\ell$ coordinates, or in other words to calculate $s$ from $\dim ({C_1})_{\bar{A}}$ we should subtract $\dim {(C_2)}_{\bar{A}}$.  
\end{proof}

Expression~(\ref{eqmaxD}) is particular well-suited for investigations in connection with schemes defined from polynomial algebras (Riemann-Roch spaces in connection with algebraic curves, polynomial rings in several variables, etc.). 
\begin{example}
Assume $\{\alpha_1, \ldots , \alpha_n\}\subseteq {\mathbb{F}}_q^m$ and let ${\mbox{ev}}: {\mathbb{F}}_q [ X_1, \ldots , X_m ]\rightarrow {\mathbb{F}}_q^n$ be given by ${\mbox{ev}}(F)=(F(\alpha_1), \ldots , F(\alpha_n))$. Further, write $C_2={\mbox{Span}}_{\mathbb{F}_q} \{ {\mbox{ev}}(M_1), \ldots , {\mbox{ev}}(M_{k_2})\}$ and  $C_1={\mbox{Span}}_{\mathbb{F}_q} \{ {\mbox{ev}}(M_1), \ldots , {\mbox{ev}}(M_{k_1})\}$, where the $M_i$'s are monomials with $M_a \prec M_b$ for $a <b$ and where $\prec$ is a fixed monomial ordering, and where we assume $\dim C_2=k_2 < \dim C_1=k_1$. Then for a given $A=\{i_1 < \cdots < i_x \}$ by~(\ref{eqmaxD}) the information held by the set of corresponding participants equal $\ell$ minus the maximal number of polynomials having different leading monomials from $\{M_{k_2+1}, \ldots , M_{k_1}\}$ with support in $\{M_1, \ldots , M_{k_1}\}$ and having $\alpha_{i_1}, \ldots , \alpha_{i_x}$ as common roots. 
\end{example}
In the next two sections we shall pursue this particular example and elaborate further on it, and by using simple methods from the theory of multivariate polynomials obtain interesting results. 

\begin{remark}
A set of participants $i_1< \cdots <i_m$ can determine all possible secrets in accordance with their shares by first determining $\{ \vec{m} \mid \vec{m} G^\prime =(c_{i_1}^\prime, \ldots , c_{i_m}^\prime)\}$ using simple linear algebra. Then they disregard the first $k_2$ coordinates of each vector, and finally remove duplicates. In particular it is clear that information quantifies in integers times $\log_2(q)$ bits, and that perfect linear ramp secret sharing schemes are exactly those with $\ell=1$.
\end{remark}

From Theorem~\ref{theuncertain} it was concluded in~\cite{Chen}[Cor.\ 4] that full privacy is ensured when $\# A < d (C_2^\perp)$ and that full recovery is guaranteed when $\# A > n - d (C_1)$. I.e.\ $t \geq d(C_2^\perp)-1$ and $r\geq n - d (C_1)+1$. Recall from~\cite{geil2014relative}[Def.\ 2] the following more refined parameters

\begin{definition}
A ramp secret sharing scheme is said to have $(t_1, \ldots , t_\ell)$-privacy and $(r_1, \ldots , r_\ell)$-reconstruction if for $m=1, \ldots , \ell$ $t_m$ is largest possible and $r_m$ is smallest possible such that
\begin{itemize}
\item no set of $t_m$ participants can recover $m$ times $\log_2(q)$ bits of information about $\vec{s}$
\item any set of $r_m$ participants can recover $m$ times $\log_2(q)$ bits of information about $\vec{s}$.
\end{itemize}
Clearly, $t=t_1$ and $r=r_\ell$.
\end{definition}

Starting in~\cite{Chen,Bains,MR2588125,kurihara} and concluding with~\cite{geil2014relative} such parameters were exactly described in terms of relative generalized Hamming weights, the proofs taking the departure in Theorem~\ref{theuncertain}, i.e.\ ~\cite{Chen}[Th.\ 10], and involving the concept of the relative dimension/length profile (RDLP) as well as that of mutual information.

\begin{definition}
For a set of nested linear codes $C_2\subseteq C_1 \subseteq {\mathbb{F}}_q^n$, and for $t=1, \ldots , \ell=\dim C_1 - \dim C_2$ the $t$th relative generalized Hamming weight is 
$$M_t(C_1,C_2)=\min \{ \# {\mbox{Supp}}(D) \mid D \subseteq C_1, D \cap C_2 = \{\vec{0}\}, \dim D= t\}.$$
\end{definition}

\begin{theorem}\label{thegen}
Consider a linear ramp secret sharing scheme defined from nested codes $C_2\subseteq C_1 \subseteq {\mathbb{F}}_q^n$ of codimension $\ell$. The privacy and reconstruction numbers satisfy
$t_m=M_m(C_2^\perp, C_1^\perp)-1$, $r_m=n-M_{\ell-m+1}(C_1,C_2)+1$ for $m=1, \ldots , \ell$.
\end{theorem}

For the sake of self containment we provide in Appendix~\ref{secproofof} a simple and direct proof avoiding the use of RDLP as well as that of mutual information.\\

By the very definition of $r_i$, for $i=1, \ldots , \ell$, the largest sets $A \subseteq I$, that do not reveal $i$ times $\log_2(q)$ bits of information, are of size exactly $r_i-1=n-M_{\ell -i+1}(C_1, C_2)$. When such sets have a systematic structure they may provide a second layer of security, which we in the next two sections shall demonstrate to be the case for families of secret sharing schemes based on monomial-Cartesian codes.

\begin{definition}
A set $A \subseteq I$ is called {\it{non-$i$-qualifying}} if from the entries corresponding to $A$ one is not able to recover $i$ times $\log_2(q)$ bits of information. The sets of largest possible size among the non-$i$-qualifying sets are called {\it{maximum non-$i$-qualifying}}. Equivalently, $A$ is maximum non-$i$-qualifying if $A \in A_{i-1}$ and $\#A=n-M_{\ell-i+1}(C_1,C_2)$.
\end{definition}

\section{Schemes based on monomial-Cartesian codes}\label{secpolsss}

In~\cite{MR3782266} two families of schemes based on monomial-Cartesian codes were shown to have significantly better parameters $t$ and $r$ than what can be produced by considering more naive coset constructions over the same polynomial algebra, e.g.\ nested generalized Reed-Muller codes. These improved schemes either have relatively large $\ell$ (\cite[Sec.\ 3]{MR3782266}) or have relatively small $\ell$ (\cite[Sec.\ 4]{MR3782266}). In the present section we show that the latter family supports the second layer of security alluded at in Section~\ref{secintro} and in Section~\ref{secsss}. Inspired by this insight in the section to follow we shall present a novel (a third) construction of so-called monomial-Cartesian ramp schemes of type C which are schemes designed to have a good second layer of security. Such schemes can have $\ell$ of both relatively small, relatively medium or relatively large size.\\

We start by recalling some results from~\cite{geil2017relative,MR3782266}. Consider a general Cartesian product point set $$S=S_1 \times \cdots \times S_m =\{\alpha_1, \ldots , \alpha_n\} \subseteq {\mathbb{F}}_q^m,$$ and write $s_i=\# S_i$. Clearly, $n=s_1 \cdots s_m$.   It is well-known that the vanishing ideal of $S$ becomes $I=\langle G_1(X_1), \ldots , G_m(X_m)\rangle \subseteq {\mathbb{F}}_q[X_1, \ldots  X_m]$ where $G_i(X_i)=\prod_{a \in S_i} (X_i-a)$ for $i=1, \ldots , m$, and we write $R={\mathbb{F}}_q[X_1, \ldots , X_m]/I$. The evaluation map ${\mbox{ev}}: R \rightarrow {\mathbb{F}}_q^n$ given by ${\mbox{ev}}(F+I)=(F(\alpha_1), \ldots , F(\alpha_n))$ is a homomorphism and when restricted to 
$${\mbox{Span}}_{\mathbb{F}_q} \{X_1^{i_1} \cdots X_m^{i_m}+I \mid 0 \leq i_v < s_v, v=1, \ldots , m\}$$ it becomes a vectorspace isomorphism. Given an arbitrary (but fixed) monomial ordering $\prec$ we write
$$\Delta(s_1, \ldots , s_m)=\{ X_1^{i_1} \cdots X_m^{i_m} \mid 0 \leq i_v < s_v, v=1, \ldots , m\} =\{N_1, \ldots , N_n\}$$
where the enumeration is according to $\prec$. Obviously, then $$\{{\mbox{ev}}(M+I) \mid M \in \Delta(s_1, \ldots ,s_n)\}$$ constitutes a basis for ${\mathbb{F}}_q^n$. The codes we consider are of the form 
$$C(L)=C(I,L)={\mbox{ev}}({\mbox{Span}}_{\mathbb{F}_q} \{ M+I \mid M \in L\}),$$ where $L \subseteq \Delta(s_1, \ldots , s_m)$ and therefore $\dim C(L)=\#L$. The notation $C(I,L)$ is in accordance with~\cite{lax}, but we shall in the remainder of the paper simply write $C(L)$ as $I$ is always clear from the context. Given a set of polynomials $\{F_1, \ldots , F_s\}$ with support in $\Delta(s_1, \ldots , s_m)$ write ${\mbox{lm}}(F_j)=X_1^{i_1^{(j)}} \cdots X_m^{i_m^{(j)}}$, $j=1, \ldots , s$, which we assume are pairwise different. The size of the support of ${\mbox{Span}}_{\mathbb{F}_q}\{ {\mbox{ev}}(F_1+I), \ldots , {\mbox{ev}}(F_s+I)\}$ is at least equal to
\begin{eqnarray}
&&\# \{ M \in \Delta(s_1, \ldots , s_m) \mid M {\mbox{ is divisible by some }} X_1^{i_1^{(j)}} \cdots X_m^{i_m^{(j)}} \}\label{eqegen} \\
&=&\# \big( \cup_{j=1}^{s} \{ X_1^{h_1} \cdots X_m^{h_m} \mid i_1^{(j)} \leq h_1 \leq s_1, \ldots , i_m^{(j)} \leq h_m \leq s_m\} \big). \nonumber 
\end{eqnarray}
 This fact corresponds to a particular incidence of the footprint bound (\cite{geilhoeholdt}), namely~\cite[Cor.\ 1]{geil2017relative} which is stated for the case $s_1=\cdots =s_m=q$, but which immediately carries over to the general case. For a different and later proof of~(\ref{eqegen}) see~\cite[Eq.\ 4]{beelen2018generalized} or \cite[Eq.\ 5]{datta2020relative}. The bound~(\ref{eqegen}) is sharp in the sence that if we write $S_i=\{\beta_1^{(i)}, \ldots , \beta_{s_i}^{(i)}\}$ then for the set of polynomials of the form 
 \begin{equation}\label{eqequal}
 F_j=\prod_{v=1}^{i_1^{(j)}}(X_1-\beta_v^{(1)}) \cdots \prod_{v=1}^{i_m^{(j)}}(X_m-\beta_v^{(m)}), j=1, \ldots , s
 \end{equation} 
 equality holds regarding the predicted support size.
  
Given $L_2 \subseteq L_1\subseteq \Delta(s_1, \ldots , s_m)$, to estimate the relative generalized Hamming weights $M_v(C,L_1),C(L_2))$ one can apply~(\ref{eqegen}), but to estimate  $M_v(C(L_2)^\perp, C(L_1)^\perp)$ one needs the Feng-Rao bound for dual codes. We collect such information in Theorem~\ref{theour} below, but first we introduce two functions.

\begin{definition} 
$$D(X_1^{i_1} \cdots X_m^{i_m})=\prod_{t=1}^{m}(s_t-i_t) {\mbox{ and }}
D^\perp(X_1^{i_1} \cdots
X_m^{i_m})=\prod_{t=1}^m(i_t+1), $$ 
and more generally, for $K
\subseteq \Delta(s_1, \ldots , s_m)$ 
\begin{eqnarray}
D(K)&=&\# \{ N \in \Delta(s_1, \ldots , s_m) \mid N {\mbox{ is divisible by some }} M \in
        K\}, \nonumber \\
D^\perp(K)&=&\# \{N \in \Delta (s_1, \ldots , s_m) \mid N {\mbox{ divides
                              some }} M \in K\}.\nonumber
\end{eqnarray}
\end{definition}

The following theorem corresponds to~\cite{MR3782266}[Th.\ 16]. 

\begin{theorem}\label{theour}
Let $S=S_1 \times \cdots \times S_m
\subseteq {\mathbb{F}}_q^m$ as above, and consider $L_2 \subset L_1 \subseteq \Delta
(s_1, \ldots , s_m)$. The codes $C(L_1)$ and $C(L_2)$ are of length $n$ and
the codimension equals
$\ell=\#L_1-\#L_2$. For $v=1, \ldots , \# L_1 - \# L_2$ we
have
\begin{eqnarray}
&M_v(C(L_1),C(L_2))\geq  \min \{{\mbox{$D$}}(K) \mid K \subseteq \{ N_u,
                          \ldots , N_n\} \cap L_1, \#K=v\},
                          \label{bnd}\\
&M_v(C(L_2)^\perp,C(L_1)^\perp)\geq  \min \{
                                       D^\perp(K)
                                      \mid K \subseteq \{N_1, \ldots
                                      N_{u^\perp} \} \backslash L_2, 
                                      \#K=v   \}, \label{bndperp} \
\end{eqnarray}
where $u=\min \{i \mid N_i \in L_1 \backslash L_2\}$ and $u^\perp =
\max \{i \mid N_i \in L_1 \}$.
\end{theorem}

To establish bounds on the numbers $t_1, \ldots , t_\ell$ and $r_1, \ldots r_\ell$, respectively, for the secret sharing scheme defined from $C_2=C(L_2) \subseteq C_1=C(L_1)$, by Theorem~\ref{thegen} we only need to apply~(\ref{bndperp}) and~(\ref{bnd}), respectively. \\

Furthermore, when given $A \subseteq {\mathcal{I}}$, according to Theorem~\ref{theuncle}, the amount of possible secrets corresponding to a given share vector $(c^\prime_{i_1}, \ldots , c^\prime_{i_{\# A}})$ of $A \subseteq {\mathcal{I}}$ equals $q^s$ where
\begin{eqnarray}
s&=&\max \{ \nu \mid  F_1(\vec{a})=\cdots = F_\nu(\vec{a})=0 {\mbox{ for all }} {\vec{a}} \in A, \nonumber \\
&&{\mbox{ \ \hspace{0.9cm}   \  \ }} {\mbox{lm}}(F_1) \prec \cdots \prec {\mbox{lm}}(F_\nu),\nonumber \\
&&{\mbox{ \ \hspace{0.9cm}   \ }}  {\mbox{ and for }} i=1, \ldots , \nu {\mbox{ it holds that }} {\mbox{Supp}}(F_i) \subseteq L_1,\nonumber \\
&&{\mbox{ \ \hspace{0.9cm}   \ }} {\mbox{ but for any non-trivial linear combination }} F {\mbox{ of }} F_1, \ldots , F_\nu \nonumber \\
&& {\mbox{ \ \hspace{0.9cm}   \ }}{\mbox{ it holds that }} M \notin L_2 {\mbox{ for some }} M \in {\mbox{Supp}}(F) \} \label{equncleeval}
\end{eqnarray}
Here, the support ${\mbox{Supp}}(F)$ of a polynomial $$F(X_1, \ldots , X_m)=\sum_{i_1, \ldots , i_m} a_{i_1, \ldots , i_m}X_1^{i_1} \cdots X_m^{i_m}$$ is the set of monomials $X_1^{i_1} \cdots X_m^{i_m}$ for which the coefficient $a_{i_1, \ldots i_m}$ is non-zero. Observe, that the latter requirement in~(\ref{equncleeval}) ensures that $$D={\mbox{Span}}_{\mathbb{F}_q} \{ {\mbox{ev}}(F_1+I), \ldots , {\mbox{ev}}(F_\nu+I)\} \cap C(I,L_2) = \{ \vec{0} \}.$$ If $L_2 \subseteq L_1$ has been chosen in such a way that all monomials in $L_1 \backslash L_2$ are larger with respect to $\prec$ than all monomials in $L_2$ then~(\ref{equncleeval}) simplifies to
\begin{eqnarray}
s&=&\max \{ \nu \mid  F_1(\vec{a})=\cdots = F_\nu(\vec{a})=0 {\mbox{ for all }} {\vec{a}} \in A, \nonumber \\
&&{\mbox{ \ \hspace{0.9cm}   \  \ }} {\mbox{lm}}(F_1) , \cdots , {\mbox{lm}}(F_\nu) {\mbox{ are pairwise different and belong to }} L_1 \backslash L_2,\nonumber \\
&&{\mbox{ \ \hspace{0.9cm}   \ }}  {\mbox{ and for }} i=1, \ldots , \nu {\mbox{ it holds that }} {\mbox{Supp}}(F_i) \subseteq L_1 \}. \label{eqsimpler}
\end{eqnarray}
The construction of the good family of schemes in~\cite[Sec.\ 4]{MR3782266} requires that $L_2=\{N_1, \ldots , N_{\# L_2}\}$ and that $L_1=\{N_1, \ldots  , N_{\# L_1}\}$ which in particular implies that we can apply~(\ref{eqsimpler}). The idea behind the construction is that in~(\ref{bnd}) as well as in~(\ref{bndperp}) one only needs to consider $K \subseteq L_1 \backslash L_2$, and by doing so, one can control the parameters $t_1, \ldots , t_\ell$ and $r_1, \ldots , r_\ell$. Furthermore, by the very definition of a monomial ordering this implies that for any $N\in L_1 \backslash L_2$ (in general for any $N \in L_1$) one has that all divisors of $N$ belong to $L_1$. This implies the existence of polynomials of the form~(\ref{eqequal}) having $N$ as leading monomial and with the support being contained in $L_1$, and thereby that the estimate on $r_i$, $i=1, \ldots , \ell$ is sharp.  We start by giving an example.

\begin{example}\label{ex1}
In this example we consider an incidence of the family of schemes with relatively small $\ell$ treated in~\cite[Sec.\ 4]{MR3782266}.
Consider $S_1, S_2 \subseteq {\mathbb{F}_q}$ where $\#S_1=\#S_2=6$ (and consequently $q \geq 7$). Choose $\prec$ to be the graded lexicographic ordering on the monomials in two variables. I.e. $X_1^{i_1}X_2^{i_2} \prec X_1^{j_1}X_2^{j_2}$ if $i_1+i_2 < j_1+j_2$, or if $i_1+i_2=j_1+j_2$, but $j_1<j_2$.  Enumerating the elements of $\Delta(s_1, s_2)$ according to $\prec$ and choosing $L_2=\{N_1, \ldots , N_{17} \}$ and $L_1$ as $L_2 \cup \{N_{18}, N_{19}\}$ the situation is as in Figure~\ref{figaro}. 
\begin{figure}[h]
$$
\begin{array}{cccccc}
N_{21}&N_{26}&N_{30}&N_{33}&N_{35}&N_{36}\\
{\text{\rund{$N_{15}$}}}&N_{20}&N_{25}&N_{29}&N_{32}&N_{34}\\
{\text{\rund{$N_{10}$}}}&{\text{\rund{$N_{14}$}}}&{\text{\rectangled{$N_{19}$}}}&N_{24}&N_{28}&N_{31}\\
{\text{\rund{$N_6$}}}&{\text{\rund{$N_9$}}}&{\text{\rund{$N_{13}$}}}&{\text{\rectangled{$N_{18}$}}}&N_{23}&N_{27}\\
{\text{\rund{$N_3$}}}&{\text{\rund{$N_5$}}}&{\text{\rund{$N_8$}}}&{\text{\rund{$N_{12}$}}}&{\text{\rund{$N_{17}$}}}&N_{22}\\
{\text{\rund{$N_1$}}}&{\text{\rund{$N_2$}}}&{\text{\rund{$N_4$}}}&{\text{\rund{$N_7$}}}&{\text{\rund{$N_{11}$}}}&{\text{\rund{$N_{16}$}}}
\end{array}
$$
\caption{The situation in Example~\ref{ex1}: The monomial in position $(i,j)$ equals $X_1^i X_2^j$.  $L_2$ corresponds to the circled monomials and $L_1$ equals $L_2$ plus the boxed monomials.}
\label{figaro}
\end{figure}
\begin{figure}[h]
$$
\begin{array}{cccccc}
6&5&4&3&2&1\\
12&10&8&6&4&2\\
18&15&12&9&6&3\\
24&20&16&12&8&4\\
30&25&20&15&10&5\\
36&30&24&18&12&6
\end{array}  {\mbox{ \ \ \ \ \ \ \ }}\begin{array}{cccccc}
6&12&18&24&30&36\\
5&10&15&20&25&30\\
4&8&12&16&20&24\\
3&6&9&12&15&18\\
2&4&6&8&10&12\\
1&2&3&4&5&6
\end{array}
$$
\caption{In the array on the left hand side we have $D(N)$ and on the right hand side $D^\perp(N)$ for all $N\in \Delta(6,6)$.}
\label{figaro2}
\end{figure}
The scheme clearly has $n=36$ participants and operates with secrets in ${\mathbb{F}}_q^{\ell=2}$, and it is clear that 
\begin{eqnarray}
t_1&=&M_1(C_2^\perp,C_1^\perp)-1 \geq \min \{ D^\perp(N_{18}), D^\perp(N_{19})\}-1=11\nonumber \\
t_2&=&M_2(C_2^\perp , C_1^\perp)-1 \geq D^\perp (\{N_{18}, N_{19}\})-1=15-1=14 \nonumber \\
r_1& =&n-M_2(C_1,C_2)+1= 36- D(\{N_{18}, N_{19} \}) +1 =36-15+1=22\nonumber \\ 
r_2&=&n-M_1(C_1,C_2)+1= 36-\min\{ D(N_{18}), D(N_{19}) \} +1=36-12+1=25. \nonumber
\end{eqnarray}
Note, that the choice of $L_1 \backslash L_2=\{N_{19}, N_{18} \}$ ensures good parameters of the scheme as both $D^\perp (N_{19})$ and $D^\perp (N_{18})$ are strictly larger than  $D^\perp(X_1^iX_2^j)$ for any other monomial on the diagonal $i+j=5$. The same thing holds for the function $D$. In this way the values of $t_i$ and $r_i$ are simultaneously optimized.\\
Now enumerate $S_\nu=\{\beta_1^{(\nu)}, \ldots , \beta_6^{(\nu)}\}$, $\nu=1,2$, and identify $\beta_s^{(\nu)}$, $s=1, \ldots , 6$ with an organization that we shall denote $O_s^{(\nu)}$. I.e.\ in total we have $2 \cdot 6 =12$ organizations $O_1^{(1)}, \ldots , O_6^{(1)}, O_1^{(2)}, \ldots ,O_6^{(2)}$,  and each  participant $(\beta_i^{(1)}, \beta_j^{(2)})$ is a representative of exactly the two organizations $O_i^{(1)}$ and $O_j^{(2)}$. Consider the polynomials $F_1=\prod_{i=1}^{3}(X_1-\beta_i^{(1)})\prod_{j=1}^{2}(X_2-\beta_j^{(2)})$ and $F_2=\prod_{i=1}^{2}(X_1-\beta_i^{(1)})\prod_{j=1}^{3}(X_2-\beta_j^{(2)})$ which have different leading monomials both belonging to $L_1 \backslash L_2$ and with all monomials in their support belonging to $L_1$. The set of non-roots of $F_1$ are $T_1= \{ \beta_4^{(1)}, \ldots , \beta_6^{(1)}\}\times \{\beta_3^{(2)},\ldots , \beta_6^{(2)}\}$ and similarly the non-roots of $F_2$ are $T_2=\{ \beta_3^{(1)}, \ldots , \beta_6^{(1)}\}\times \{\beta_4^{(2)},\ldots , \beta_6^{(2)}\}$. Hence, if $A$ does {\textit{not}} contain an element from $T_1$ then $F_1$ has all of $A$ as roots, and similarly if $A$ does {\textit{not}} contain an element from $T_2$ then $F_2$ has all of $A$ as roots. In both cases $A$ can recover at most $1$ times $\log_2(q)$ bits of information. Both sets are of size $r_\ell-1=24$ and therefore maximum non-$2$-qualifying. This describes some kind of considerate property in that by leaving out representatives from any $3$ out of the first type of organizations who are simultaneously  representatives from any fixed $4$ out of the second type of organizations one cannot recover the entire secret, and similarly with $4$ out of the first type and $3$ out of the second type. \\
The elements which are {\textit{not}} common roots of $F_1$ and $F_2$ are $T_1 \cup T_2 $. Hence, if $A$ does {\textit{not}} contain an element from $T_1 \cup T_2$ then $A$ cannot recover any information. By inspection such $A$ is maximum non-$1$-qualifying, i.e. they are of size $r_1-1=21$. Again we can interpret this as some kind of considerate property, in that given any set of $4$ organizations of the first type and any set of $4$ organizations of the second type one cannot leave out more than one participant representing one from each set if one wants to recover any information. \\
As the established maximum non-$i$-qualifying sets have a nice systematic structure and are of size significantly larger than the bound on $t_i$ we have established a second layer of security. 
\end{example}

In the above example we considered a case with two variables and $s_1=s_2$. For such case choosing the graded lexicographic ordering is the optimal choice as along the diagonal 
$$\{X_1^{i_1}X_2^{i_2} \mid i_1+i_2=e\}$$
where, $e$ is any fixed integer in $[0, 2 (q-1)]$, and $q$ is the field size, the values of both $D$ and $D^\perp$ are highest possible at the center, and decreases symmetrically the larger the distance is to the center (i.e. the further $i_1$ becomes from $i_2$).  We now state a theorem describing such case in general. The results regarding $r_i$, and $t_i$ are a direct adaption of~\cite[Th.\ 27]{MR3782266}, whereas the treatment of maximum non-$i$-qualifying sets with the corresponding kind of considerate property is new.
\begin{theorem}\label{the27ish}
Consider $S_1, S_2 \subseteq {\mathbb{F}}_q$ with $s=\#S_1=\#S_2$. Let  $\prec$ be the graded lexicographic ordering with $X_1 \prec X_2$. Consider $0 \leq i_1 < i_2 \leq s-1$ and let
\begin{eqnarray}
L_1&=&\{ N \in \Delta(s,s) \mid N \preceq X_1^{i_1}X_2^{i_2} \} \nonumber \\
L_2&=&\{ N \in \Delta(s,s) \mid N \preceq X_1^{i_2}X_2^{i_1} \}. \nonumber
\end{eqnarray} 
The secret sharing scheme defined from $C(L_2) \subseteq C(L_1)$ has parameters $\ell = i_2-i_1+1$ and for $m=1, \ldots , \ell$
\begin{eqnarray}
t_m&\geq& (i_1+m)(i_2+1)-\dfrac{m(m-1)}{2} -1  \nonumber \\
r_m&=&s^2-\left[(s-i_1)(s-i_2)+\sum_{t=2}^{\ell-m+1}((s-i_1)-(t-1))\right]+1 \label{eqmidle} \\
&=& s^2-(s-i_1)(s-i_2+\ell-m) +\dfrac{(\ell-m+1)(\ell-m)}{2}+1   \nonumber
\end{eqnarray}
Write $S_\nu =\{\beta_1^{(\nu)}, \ldots , \beta_s^{(\nu)}\}$, for $\nu=1,2$, and $S=S_1 \times S_2$. 
For $w=1, \ldots , \ell$ define
$$T_w=\{ \beta_{(i_2+2)-w}^{(1)}, \ldots , \beta_s^{(1)}\} \times \{\beta_{i_1+w}^{(2)}, \ldots , \beta_s^{(2)}\}.$$
Then for any set ${\mathcal{J}} \subseteq \{1, \ldots , \ell\}$ it holds that 
\begin{equation}
S \backslash \cup_{i \in {\mathcal{J}}}T_i \label{eqcupcake}
\end{equation}
is a non-$u$-qualifying set
where $u=\ell+1-\# {\mathcal{J}}$. When ${\mathcal{J}}=\{ 1, \ldots , \# {\mathcal{J}}\}$ or ${\mathcal{J}}=\{u, \ldots , \ell\}$ then~(\ref{eqcupcake}) is a maximum non-$u$-qualifying set.
\end{theorem}
\begin{proof}
The parameters are established by combining~\cite[Th.\ 27]{MR3782266} with Theorem~\ref{thegen}. That~(\ref{eqcupcake}) is a non-$u$-qualifying set follows from similar arguments as in Example~\ref{ex1}. The result regarding maximum non-$u$-qualifying sets follows from the fact that $\# (T_1 \cup \cdots \cup T_{\mathcal{J}})=\#(T_u \cup \cdots \cup T_\ell)$ equals the content of the square bracket in~(\ref{eqmidle})
\end{proof}

\begin{remark}
Similar remarks concerning considerate properties as described in Example~\ref{ex1} apply to the construction in Theorem~\ref{the27ish}.
\end{remark}
The situation of three or more variables is less trivial compared to Theorem~\ref{the27ish} in which we treated the two-variable case, but still the idea behind the improved construction of nested codes can be applied in some cases. We here only treat the situation of the codimension $\ell$ being equal to $1$ which we illustrate by an example that immediately generalizes to all codimension $1$ cases. 
\begin{example}\label{ex3}
Consider $S_1, S_2, S_3 \subseteq {\mathbb{F}_q}$ with $s=\# S_1=\#S_2=\#S_3=4$. Let $\prec$ be a graded lexicographic ordering and let $L_2=\{ N \in \Delta(4,4,4) \mid N \prec X_1^2X_2^2X_3^2\}$ and $L_1=L_2 \cup \{ X_1^2X_2^2X_3^2\}$. We have $M(C(L_1),C(L_2))=D(X_1^2X_2^2X_3^2)=8$ and $M(C(L_2)^\perp, C(L_1)^\perp)\geq D^\perp(X_1^2X_2^2X_3^2)=27$. Note, that $X_1^2X_2^2X_3^2$ is the monomial with both highest $D^\perp$- and $D$-value among the monomials in $\Delta(4,4,4)$ of total degree $6$, in which way we have optimized the parameters of the scheme. Hence, the secret sharing scheme based on $C(L_2)\subseteq C(L_1)$ has $n=64$ participants, a secret of $\ell=1$ times $\log_2(q)$ bits and $t \geq 26$, $r=57$. Write $S_\nu=\{\beta_1^{(\nu)}, \beta_2^{(\nu)}, \beta_3^{(\nu)},\beta_4^{(\nu)}\}$, $\nu=1,2,3$, which we identify with $4$ different organizations $O_1^{(\nu)}, O_2^{(\nu)}, O_3^{(\nu)}, O_4^{(\nu)}$ at level $\nu$. In this way an element of $S=S_1\times S_2 \times S_3$ is a member of a unique organization at each of the three levels. Now 
$$S \backslash \big( \{\beta_3^{(1)}, \beta_4^{(1)}\} \times \{\beta_3^{(2)}, \beta_4^{(2)}\} \times \{\beta_3^{(3)}, \beta_4^{(3)}\} \big)$$
is a maximum non-$1$-qualifying set (i.e. a set of size $56$ possessing no information). The scheme can be viewed as having some kind of considerate properties in that given an organization of size $2$ at each of the three levels by leaving out all participants who simultaneously represent these three organizations one is not able to recover any information.
\end{example}

As illustrated in \cite[Ex.\ 14]{MR3782266} already for two variables optimizing simultaneously the parameters $r_1, \ldots , r_\ell$ and $t_1, \ldots , t_\ell$ when given fixed $\ell$ may not be possible when the sets $S_i$ are of different sizes. However, with the second layer of security in mind, it makes sense to concentrate mainly on the parameters $r_1, \ldots , r_\ell$ over $t_1, \ldots , t_\ell$ ensuring systematic maximum non-$i$-qualifying sets with structures similar to what is described in Theorem~\ref{the27ish}. We illustrate the idea with an example.

\begin{example}\label{ex4}
Consider $S_1, S_2 \subseteq {\mathbb{F}}_q$ with $\#S_1=8$ and $\#S_2=5$. The $D$ and $D^\perp$ values of the elements of $\Delta(8,5)$ are depicted in Figure~\ref{figaro3}. Choosing $\prec$ to be the graded lexicographic ordering with $X_1 \prec X_2$ and letting $L_2=\{N \mid N \prec X_1^5X_2\}$ and $L_1=L_2 \cup \{  X_1^5X_2, X_1^4X_2^2\}$ we locally optimize the $D$ values as $D(X_1^5X_2)=D(X_1^4X_2^2=12$ which is larger than $D(X_1^{i_1}X_2^{i_2})$ for any other $(i_1, i_2)$ with $i_1+i_2=6$. However, the $D^\perp$ values are not optimized in a similar fashion as $D^\perp (X_1^5X_2)=12$ and $D^\perp (X_1^4X_2^2)=15$, which are both smaller than $D^\perp (X_1^3X_2^3)=16$. We obtain $n=40$, $\ell=2$,  $r_1=n-D(X_1^5X_2,X_1^4X_2^2)+1=40-15+1=26$ , $r_2=40-12+1=29$, $t_1 \geq \min\{ D^\perp X_1^5X_2), D^\perp (X_1^4X_2^2)\}-1=11$, and $t_2 \geq D^\perp (X_1^5X_2X_2,X_1^4X_2^2)-1=16$. Identifying the elements of $S_1$ with $8$ different organizations $O_1^{(1)}, \ldots , O_8^{(1)}$ and similarly the element of $S_2$ with $5$ different organizations $O_1^{(2)}, \ldots , O_5^{(2)}$ by a bijective map each participant represents a unique organization from each of the two sets. Given any $3$ organizations from the first level and any $4$ organizations from the second level, by leaving out all participants representing simultaneously one organization from each subset, one is at most able to recover $1$ times $\log_2(q)$ bits of information. Similarly with $4$ organizations from the first level and $3$ from the second level. Leaving out the union of the about mentioned participants one cannot recover any information. Hence, even though the only information we have on $t_1$ is that it is at least $11$, we have a series of systematic sets of size $40-12=28$ who cannot recover more than $1$ times $\log_2(q)$ bits of information. Similarly, even though the only information we have on $t_2$ is that it is at least $15$, we have a series of systematic sets of size $40-15=25$ from who cannot reveal any information. Adapting other types of monomial ordering does not seem to help optimizing locally both $D$ and $D^\perp$.
\begin{figure}[h]
$$
\begin{array}{cccccccc}
8&7&6&5&4&3&2&1\\
16&14&12&10&8&6&4&2\\
24&21&18&15&12&9&6&3\\
32&28&24&20&16&12&8&4\\
40&35&30&25&20&15&10&5
\end{array}  {\mbox{ \ \ \ \ \ \ \ }}\begin{array}{cccccccc}
5&10&15&20&25&30&35&40\\
4&8&12&16&20&24&28&32\\
3&6&9&12&15&18&21&24\\
2&4&6&8&10&12&14&16\\
1&2&3&4&5&6&7&8
\end{array}
$$
\caption{In the array on the left hand side we have $D(N)$ and on the right hand side $D^\perp(N)$ for all $N\in \Delta(8,5)$.}
\label{figaro3}
\end{figure}
\end{example}

\begin{remark}
As Theorem~\ref{the27ish} generalizes Example~\ref{ex1} it is straightforward to generalize Example~\ref{ex4} to cover all cases of point sets $S_1 \times S_2$. We leave the details for the reader.
\end{remark}

\section{Considerate schemes}\label{secdem}

In this section we optimize the second layer of security whilst paying no interest in the worst case information leakage in terms of number of participants. I.e. we are interested in the second layer of security and in the reconstruction numbers $r_1, \ldots , r_\ell$, but downplay the interest in the numbers $t_1, \ldots , t_\ell$. Our new construction is designed to have very systematic maximum non-$u$-qualifying sets for any $u=1, \ldots , \ell$, the systematic form giving rise to considerate properties along the same line as those described in the previous section. A particular advantage of the new construction is that in contrast to the construction of~\cite[Sec.\ 4]{MR3782266} as treated in Section~\ref{secpolsss}, it allows for a great variety of possible values of $\ell$, including small, medium sized or large. The codes are defined by the same evaluation map ${\mbox{ev}}:R \rightarrow {\mathbb{F}}_q^n$ as in Section~\ref{secpolsss}, but we shall employ a different monomial ordering. We gain extra freedom by no longer requiring that all monomials in $L_2$ are smaller than the remaining monomials in $L_1$ with respect to the applied monomial ordering. To deal with this new situation we introduce the following definition.

\begin{definition}\label{defdem}
Let $S_1, \ldots , S_m \subseteq {\mathbb{F}}_q$ all of size at least $2$ and write $s_1=\# S_1, \ldots , s_m=\# S_m$. Choose integers $0 \leq v_1 < j_1 < s_1, \ldots , 0 \leq v_m < j_m < s_m$, and let 
$$L_1=\{X_1^{i_1} \cdots X_m^{i_m} \mid 0 \leq i_1 \leq j_1, \ldots , 0 \leq i_m \leq j_m\}$$ and $L_2=L_1\backslash \Box$ where $$\Box =\{X_1^{i_1} \cdots X_m^{i_m} \mid v_1 \leq i_1 \leq j_1, \ldots , v_m \leq i_m \leq j_m \}.$$ 
Consider the lexicographic ordering $\prec$ with $X_m \prec \cdots \prec X_2 \prec X_1$, and write $N_{\mbox{min}}=X_1^{v_1} \cdots X_m^{v_m}$ (which is the minimal element of $\Box$ with respect to $\prec$), and define
$$\Diamond=\{M \in L_1 \mid N_{\mbox{min}} \preceq M\}.$$
A secret sharing scheme defined from $C_2=C(L_2) \subseteq C(L_1)=C_1$ as above and satisfying that for all $i=1, \ldots , \ell=\# \Box$ it holds that among those $K \subseteq \Diamond$ of size $i$ a $K^\prime \subseteq \Box$ exists with $D(K^\prime)$ being minimal is called a monomial-Cartesian ramp scheme of type C.
\end{definition}

\begin{example}\label{exdem1}
In this example we illustrate the notation $\Box$ and $\Diamond$ from Definition~\ref{defdem} in the case of $m=2$ and $s_1=6$, $s_2=5$, $j_1=3$, $j_2=3$, $v_1=1$, and $v_2=2$.
Figure~\ref{figdiamond} illustrates the situation.
\begin{figure}[h]
$$
\begin{array}{cccccc}
\cdot&\cdot&\cdot&\cdot&\cdot&\cdot\\
\cdot&-&-&-&\cdot&\cdot \\
\cdot&-&-&-&\cdot&\cdot \\
\cdot & \cdot &+&+&\cdot &\cdot \\
\cdot & \cdot &+&+&\cdot &\cdot
\end{array}
$$
\caption{The situation in Example~\ref{exdem1}. Monomials marked with ``$-$'' correspond to $\Box$. Adding the monomials marked with ``$+$'' one obtains $\Diamond$.}
\label{figdiamond}
\end{figure}
\end{example}

\begin{remark}
For the particular case of $m=1$ the monomial-Cartesian ramp schemes of type C correspond to Blakley and Meadows enhancement of Shamir's secret sharing scheme in~\cite{Blakley} and the last condition in Definition~\ref{defdem} is trivially satisfied.  
\end{remark}

\begin{remark}\label{remdem}
When $C_1=C(L_2) \subseteq C(L_1)=C_2$ are defined as in Definition~\ref{defdem} then we only need to consider in~(\ref{equncleeval}) as leading monomials those that belong to $\Diamond$. The very last condition of Definition~\ref{defdem} then implies that we can actually apply~(\ref{eqsimpler}), as indeed $\Box=L_1 \backslash L_2$. Hence, to calculate the relative generalized Hamming weights $M_t(C_1,C_2)$ and from that the reconstruction numbers $r_1, \ldots , r_\ell$, we can employ~(\ref{eqsimpler}) although the conditions of~\cite[Sec.\ 4]{MR3782266} are not satisfied.
\end{remark}

\begin{theorem}\label{thetypeC}
Let $C(L_2) \subseteq C(L_1)$ define a monomial-Cartesian ramp scheme of type C (Definition~\ref{defdem}). We have $r_1= s_1 \cdots s_m -(s_1-v_1) \cdots (s_m-v_m)+1$  and $r_\ell = s_1 \cdots s_m - (s_1-j_1) \cdots (s_m-j_m)+1$. For $i=1, \ldots , m$ consider pointsets $T_{\max}^{(i)} \subseteq S_i$ with $\#T_{\max}^{(i)}=s_i-v_i$ and $T_{\min}^{(i)} \subseteq S_i$ with $\#T_{\min}^{(i)}=s_i-j_i$. Then $(S_1 \times \cdots \times S_m) \backslash (T_{\max}^{(1)} \times \cdots \times  T_{\max}^{(m)})$ is a maximum non-$1$-qualifying set and $(S_1 \times \cdots \times S_m) \backslash (T_{\min}^{(1)} \times \cdots \times  T_{\min}^{(m)})$ is a maximum non-$\ell$-qualifying set. The latter describes the only such sets.
\end{theorem}
\begin{proof}
The main part of the theorem follows from Remark~\ref{remdem}. To prove the last result we must  establish that the only non-zero polynomials $F(X_1, \ldots , X_m)$ with ${\mbox{Supp}}(F) \subseteq L_1$ and having exactly $D(X_1^{j_1}\cdots X_m^{j_m})$ non-roots, are all of the form
$$\prod_{\alpha_1 \in S_1 \backslash T_{\min}^{(1)}} (X_1-\alpha_i) \cdots \prod_{\alpha_m \in S_m \backslash T_{\min}^{(m)}}(X_m-\alpha_m).$$ Clearly, such a polynomial must have $X_1^{j_1} \cdots X_m^{j_m}$ as leading monomial as this monomial is the only in $L_1$ of such a small value.
By inspection the last result of the theorem is then a direct consequence of~\cite[Th.\ 7]{geil2021multivariate} (where some care must be taken as the $D$ are used in a slightly different meaning there).
\end{proof}

We next coin the concept of considerate ramp secret sharing to cover the general case where there exists systematic maximum non-$1$-qualifying sets as in Theorem~\ref{thetypeC}.

\begin{definition}\label{defconman}
Consider general linear codes $C_2 \subset C_1 \subseteq {\mathbb{F}}_q^n$ with $\ell = \dim C_1 -\dim C_2 \geq 1$ and where $n=s_1\cdots s_m$ for integers $2 \leq s_i \leq q$, $i=1, \ldots , m$. Assume the existence of integers $v_1 < s_1, \ldots , v_m < s_m$ such that $M_{\ell}(C_1, C_2)=(s_1-v_1) \cdots (s_m-v_m)$. Let $A_1, \ldots , A_m \subseteq {\mathbb{F}}_q$ of sizes $\# A_i=s_i$, $i=1, \ldots , m$ be arbitrary. If for some bijection $$b: A_1 \times \cdots \times A_m \rightarrow \{1, \ldots , n\}$$ it holds that for all sets $A_1^\prime \subseteq A_1, \ldots , A_m^\prime \subseteq A_m$ with $\# A_1^\prime = s_1-v_1, \ldots , \# A_m^\prime =s_m-v_m$ there exists a space $D$ of dimension equal to $\ell$ satisfying 
$D\subseteq C_1$
$D \cap C_2=\{ \vec{0}\}$ and 
\begin{equation*}
{\mbox{Supp}}(D) = b (A_1^\prime \times \cdots \times A_m^\prime) 
\end{equation*} then the secret sharing scheme base on $C_2 \subseteq C_1$ is said to be considerate. 
\end{definition}

So considerate ramp secret sharing schemes are schemes that protect against systematic discrimination along the lines of the present paper in connection with recovery of smallest possible non-zero amount of information and clearly monomial-Cartesian ramp schemes of type C constitute well-structured examples of such schemes. We stress that Definition~\ref{defconman} is different from the concept of hierarchical secret sharing~\cite{Tassa}.

\begin{proposition}\label{prodem}
Let the notation be as in Definition~\ref{defdem} with $s_1-j_1= \cdots =s_m-j_m$ and $j_1-v_1= \cdots =j_m-v_m$ then one obtains a monomial-Cartesian ramp scheme of type C. For $m=2$ with $s_1-j_1 \geq s_2-j_2$ with $j_1-v_1 \leq j_2-v_2+1$ one similarly obtain a monomial-Cartesian ramp scheme of type C.
\end{proposition}
\begin{proof}
We only proof the latter result. We first observe that for $D(K)$ to be minimal among those $K \in \Diamond$ with $\#K=i$ it clearly holds that if $N \in K$ then all $M \in \Diamond$ which are divisible by $N$ also belong to $K$. For a given $i$ with  $1 \leq i \leq \ell=\# \Box$ consider the possible $K \subseteq \Diamond$, $\#K=i$, with $D(K)$ being minimal. Among such sets let $K$ be chosen to contain the minimal possible number of elements  outside $\Box$. If this number equals $0$ then we are through. Hence, assume this is not the case. Without loss of generality we may assume that $K \cap \Box$ consists of full line segments $X_1^{i_1}X_2^{v_2}, X_1^{i_1}X_2^{v_2+1}, \ldots , X_1^{i_1}X_2^{j_2}$ for $i_1=a, \ldots , j_1$ plus possibly some partial line segment $X_1^{a-1}X_2^b, X_1^{a-1}X_2^{b+1}, \ldots , X_1^{a-1}X_2^{j_2}$. Here, $v_1+1 \leq a \leq j_1$ and $v_2+1 \leq b  \leq j_2$. If $K \cap \Box$ contains the mentioned partial line segment then remove from $K$ the smallest element according to $\prec^\prime$, where $\prec^\prime$ is the lexicographic ordering with $X_1 \prec^\prime X_2$. By assumption this monomial belongs to $\Diamond \backslash  \Box$. Then replace it with $X_1^{a-1}X_2^{b_1}$ to obtain a new $K$ with the same $D$-value (or smaller), but having less elements outside $\Box$. But this is in contradiction with our assumption that the original $K$ was chosen to have the smallest number of elements outside $\Box$. Hence, the partial line segment cannot exist. But then if $X_1^{c}X_2^{d}$ is the smallest element in $\Diamond \backslash \Box$ with respect to $\prec^\prime$, then by removing from $K$ the line segment $X_1^{c}X_2^{d}, X_1^{c+1}X_2^{d}, \ldots , X_1^{j_1}X_2^{d}$ and adding instead the line segment $X_1^{a-1}X_2^{j_1-j_2+c}, X_1^{a-1}X_2^{j_1-j_2+c+1}, \ldots , X_1^{a-1}X_2^{j_2}$ we obtain a new $K$ having fewer elements outside $\Box$ and being of $D$-value smaller than or equal to that of the previous $K$. Again, this is a contradiction. 
\end{proof}
\begin{example}\label{exdem2}
The secret sharing scheme with $\Box$ and $\Diamond$ as in Example~\ref{exdem1} is a monomial-Cartesian ramp scheme of type C.
\end{example}

\begin{example}\label{exdem3}
Let notation be as in Definition~\ref{defdem} with $q > 10$, and choose $S_1, S_2 \subseteq {\mathbb{F}}_q$ of equal size $s_1=s_2=10$. Let $v_1=v_2=2$ and $j_1=j_2=5$. The corresponding scheme by Proposition~\ref{prodem} is a monomial-Cartesian ramp scheme of type C, and we have $n=100$ and $\ell= 16$. Consider arbitrary $T_1^{(i)}, \ldots , T_4^{(i)} \subseteq S_i$, $i=1,2$ of corresponding sizes $5, 6, 7, 8$ (in that order). We have $M_{16}(C_1,C_2) = D(\{ X_1^{i_1}X_2^{i_2} \mid 2 \leq i_1 \leq 5, 2 \leq i_2 \leq 5 \} )=64$, $r_1=n-M_{16}(C_1,C_2)+1=37$ and $(S_1 \times S_2) \backslash (T_4^{(1)} \times T_4^{(2)})$ is a maximum non-$1$-qualifying set. Hence, leaving out all members of $T_4^{(1)} \times T_4^{(2)}$ one cannot obtain any information. We next see that $M_{16-z}(C_1, C_2)= 64-z$ for $z=1, 2, 3$. To see this note that the relative generalized Hamming weights constitute a strictly increasing sequence, and that one cannot dicrease the value of $D$ by more than the number $z$ of removed elements from $\{ X_1^{i_1}X_2^{i_2} \mid 2 \leq i_1 \leq 5, 2 \leq i_2 \leq 5 \}$ in the argument of $D$ as long as $z<4$. Next $M_{12}(C_1,C_2)=D( \{ X_1^{i_1}X_2^{i_2} \mid 3 \leq i_1 \leq 5, 2 \leq i_2 \leq 5 \} )  =  D( \{ X_1^{i_1}X_2^{i_2} \mid 2 \leq i_1 \leq 5, 3 \leq i_2 \leq 5 \} ) =M_{16}(C_1,C_2)-8=56$ (a jump in $1+4=5$ from $M_{13}(C_1, C_2)$). Continuing this way we see that removing $4+z$ from $\{ X_1^{i_1}X_2^{i_2} \mid 2 \leq i_1 \leq 5, 2 \leq i_2 \leq 5 \}$ , $z=1,2$ the $D$-value decrease by at most $4+z$, and therefore $M_{16-(4+z)}(C_1, C_2)=M_{16}-(4+z)$. Next, we can remove $4+3$ elements in such a way that the corresponding $D$ is smallest possible as follows $M_{9}(C_1, C_2)=D(\{ X_1^{i_1}X_2^{i_2} \mid 3 \leq i_1 \leq 5, 3 \leq i_2 \leq 5 \} )=M_{16}(C_1,C_2)-7-4-4=49$. Continuing this way see that $M_{9-z}(C_1,C_2)=M_9(C_1,C_2)-z$ for $z=1,2$ and that $M_6(C_1,C_2)=M_{9-3}(C_1,C_2)=M_9(C_1,C_2)-3-4=42$. Next, in a similar fashion $M_{6-1}(C_1,C_2)=M_6(C_1,C_2)-1=41$, but $M_{6-2}(C_1, C_2)=M_{6}-2-4= 36$. Finally, $M_{3}(C_1, C_2)=M_{4}(C_1,C_2)-1=35$, $M_{2}(C_1,C_2)=M_4(C_1,C_2)-2-4=30$, and $M_1(C_1,C_2)=M_2(C_1,C_2)-1-4=25$. Regarding maximum non-$i$-qualifying sets we have the following picture. Including from $T_4^{(1)} \times T_4^{(2)}$, exactly $z$ members one can at most obtain $z$ times $\log_2(q)$ bits of information, $z=1,2,3$. Including no members of $T_3^{(1)}\times T_4^{(2)}$ or $T_4^{(1)} \times T_3^{(2)}$ one can obtain (at most) $4$ times $\log_2(q)$ bits of information. Including from $T_3^{(1)}\times T_4^{(2)}$ or $T_4^{(1)} \times T_3^{(2)}$ at most $z=1,2$ members one can at most obtain $4+z$ times $\log_2(q)$ bits of information. Including no members of $T_3^{(1)} \times T_3^{(2)}$ one can at most obtain $7$ times $\log_2(q)$ bits of information. Including from $T_3^{(1)} \times T_3^{(2)}$ at most $z=1,2$ members one can at most obtain $7+z$ times $\log_2(q)$ bits of information. Including no elements from $T_2^{(1)} \times T_3^{(2)}$ or $T_3^{(1)} \times T_2^{(2)}$ one can obtain (at most) $10$ times $\log_2(q)$ bits of information. Including at most one element from either of these sets, one can at most obtain $11$ times $\log_2(q)$ bits of information. Including no elements from $T_2^{(1)} \times T_2^{(2)}$ one can (at most) obtain $12$ times $\log_2(q)$ bits of information. Including at most one element from  $T_2^{(1)} \times T_2^{(2)}$ one can (at most) obtain $13$ times $\log_2(q)$ bits of information. Including no elements from $T_1^{(1)} \times T_2^{(2)}$ or $T_2^{(1)} \times T_1^{(2)}$ one can (at most) obtain  $14$ times $\log_2(q)$ bits of information. Finally, including no elements from $T_1^{(1)} \times T_1^{(2)}$ one can (at most) obtain $15$ times $\log_2(q)$ bits of information. 
\end{example}
\begin{example}\label{exdem4}
This is a continuation of Example~\ref{exdem3}. Let instead $v_1=v_2=3$ and $j_1=j_2=6$ which again gives us a monomial-Cartesian ramp scheme of type C. Let the sizes of $T_1^{(i)}, T_2^{(i)}, T_3^{(i)}$, and $T_4^{(i)}$ be  $4, 5, 6$ and $7$ (in that order) for $i=1,2$. Again we obtain maximum non-$16$-qualified, maximum non-$13$-qualified,maximum non-$8$-qualified and maximum non-$1$-qualified sets as described in the previous example. By leaving out Cartesian product pointsets of size $4 \cdot 4$, $5 \cdot 5$, $6 \cdot 6$, and $7 \cdot 7$, respectively, one is not able to obtain all information, $13$ times $\log_2(q)$ bits of information, $8$ times $\log_2(q)$ bits of information and $1$ times $\log_2(q)$ bits of information, respectively. Now increasing all of  $v_1, v_2, j_1, j_2$ by $1$ one cannot leave out any Cartesian product pointset of size $3 \cdot 3$, $4 \cdot 4$, $5 \cdot 5$, and $6\cdot 6$, respectively, if one wants to obtain the mentioned amount of information. Increasing, again the parameters by $1$, one cannot leave out any Cartesian product pointsets of size $2 \cdot 2$, $3 \cdot 3$, $4 \cdot 4$, and $5 \cdot 5$, respectively.

\end{example}

\begin{theorem}\label{thedem1}
Consider a monomial-Cartesian ramp scheme of type C as in Definition~\ref{defdem} with $m=2$ (two variables), $s=s_1=s_2$, $v=v_1=v_2$, and $j=j_1=j_2$. We have $n=s^2$ participants and the secret is of size $\ell=(j-v)^2$. For $\ell \geq e \geq 1$ one can write $e$ uniquely in one of the following ways
\begin{eqnarray}
e=z(z-1)-h & {\mbox{ with }} & \sqrt{\ell} \leq z \leq 2 {\mbox{ and }} 0 \leq h < z-1 \label{eqdem1}
\end{eqnarray}
or
\begin{eqnarray}
e=z^2-h & {\mbox{ with }} &\sqrt{\ell} \leq z \leq 1 {\mbox{ and }} 0  \leq h \leq z-1 \label{eqdem2}
\end{eqnarray}
In situation~(\ref{eqdem1}) we have
$$M_e(C_1, C_2)=(s-j-1+z)(s-j-2+z) -h$$
$$r_{\ell - e+1} = s^2-(s-j-1+z)(s-j-2+z)+h+1$$
and by excluding any Cartesian product pointset of size $z \times (z-1)$ or $(z-1) \times z$, but up till $h$ elements herein, one obtains at most $\ell -e$ times  $\log_2(q)$ bits of information. In situation~(\ref{eqdem2}) we have
$$M_e(C_1,C_2)=(s-j-1+z)^2-h$$
$$r_{\ell -e+1}=s^2-(s-j-1+z)^2-h$$
and by excluding any Cartesian product pointset of size $z \times z$, but up till $h$ elements herein, one obtains at most $\ell-e$ times $\log_2(q)$ bits of information. In both situations $S_1 \times S_2$ with such a set removed constitutes a maximum non-$(\ell - e)$-qualifying set. (if clever points are removed, otherwise not even that)
\end{theorem}
\begin{proof}
The proof uses similar arguments as in Example~\ref{exdem3}.
\end{proof}

In a straightforward way one can generalize Theorem~\ref{thedem1} to arbitrary $m \geq 2$. For simplicity we here only treat the cases where entire Cartesian product pointsets are excluded.

\begin{theorem}\label{thedem2}
Consider a monomial-Cartesian ramp scheme of type C as in Definition~\ref{defdem} with arbitrary $m \geq 2$. Assume $s=s_1=\cdots =s_m$, $v= v_1 = \cdots =v_m$ and $j=j_1= \cdots =j_m$. We have $s^m$ participants and the secret is of size $\ell =(j-v)^m$. Assume $e=z^a(z-1)^{b}$ for some $a+b=m$ and 
$j-v \leq z \leq 1$ if $b=0$ or $j-v \leq z \leq 2$ if $b >0$. Then 
$$M_e(C_1, C_2) = (s-j-1+z)^a(s-j-2+z)^b$$
$$r_e(C_1, C_2) =s^m-(s-j-1+z)^a(s-j-2+z)^b+1$$
and by excluding any Cartesian product pointset of size $f_1 \times \cdots \times f_m$ where exactly $a$ of the $f_i$'s equal $s-j-1+z$ and the remaining equal $s-j-2+z$ one obtains at most $\ell -e$ times $\log_2(q)$ bits of information. $S_1 \times \cdots \times S_m$ with such a set removed constitutes a maximum non-$(\ell - e)$-qualifying set.
\end{theorem}

\begin{remark}
Theorem~\ref{thedem1} and Theorem~\ref{thedem2} can in a straight forward manner be generalized to treat the first general case described in Proposition~\ref{prodem}, by noting that $s_1-j_1=\cdots =s_m-j_m$ and thereby corresponds to the $s-j$, and therefore the relative generalized Hamming weights are the same. The only changes needed is to replace $s^m$ with $s_1 \cdots  s_m$.
\end{remark}

We finally treat a particular secret sharing scheme fulfilling the last mentioned requirements of Proposition~\ref{prodem}

\begin{example}
This is a continuation of Example~\ref{exdem1} and Example~\ref{exdem2} where we treated an incidence of the second construction mentioned in Proposition~\ref{prodem}. We have 
\begin{eqnarray}
M_6&=&  D(\Box)=16 \nonumber \\
M_5&=&  D(\Box \backslash \{X_1X_2^2\})=15\nonumber \\
M_4&=&  D(X_1^2X_2^2, X_1^2X_2^3,X_1^3X_2^2,X_1^3X_2^3)=14 \nonumber \\
M_3&=& D(X_1X_2^3,  X_1^2X_2^3, X_1^3X_2^3   )=10 \nonumber \\
M_2&=& D(X_1^2X_2^3, X_1^3X_2^3)=8 \nonumber \\
M_1&=& D(X_1^3X_2^3)=6 \nonumber 
\end{eqnarray}
Let $T_1^{(1)}, T_2^{(1)}, T_3^{(1)}$, respectively, be a subset of $S_1$ of cardinality $3, 4, 5$, respectively, and let $T_1^{(2)},T_2^{(2)}$, respectively, be a subset of $S_2$ of cardinality 2, 3, respectively. Then $(S_1 \times S_2) \backslash (T_3^{(1)} \times T_2^{(2)})$ is a maximum non-$1$-qualifying set, and by adding any extra element one obtains a maximum non-$2$-qualifying set. Further $(S_1\times S_2) \backslash (T_2^{(1)} \times T_2^{(2)})$ is a maximum non-$3$-qualifying set, $(S_1 \times S_2) \backslash (T_3^{(1)} \times T_1^{(2)})$ is a maximum non-$4$-qualifying set, $(S_1\times S_2)\backslash (T_2^{(1)}\times T_1^{(2)})$ is a maximum non-$5$-qualifying set, and finally $(S_1\times S_2) \backslash (T_1^{(1)}\times T_1^{(2)})$ is a maximum non-$6$-qualifying set. Observe that a maximum non-$4$-qualifying set cannot be expanded to a maximum non-$3$-qualifying set. This is in contrast to the situation for the first construction mentioned in Proposition~\ref{prodem} where similar inclusions are always be possible.
\end{example}

\section{Concluding remarks}\label{secconcl}
We leave it as an open research problem to describe additional considerate ramp secret sharing schemes to the monomial-Cartesian ramp schemes of type C and to investigate their second layer security with respect leakage of $2$ times $\log_2(q)$ or more bits of information. We also leave it as an open research problem to determine other monomial-Cartesian ramp Schemes of type C than those covered by Proposition~\ref{prodem} and to establish for such schemes information on maximum non-$u$-qualifying sets. 

In another direction we propose to apply Theorem~\ref{theuncle} to obtain maximum non-$i$-qualifying sets of schemes defined from nested one-point algebraic geometric codes. The graded structure behind such codes are well-suited for calculating the numbers~(\ref{eqmaxD}) and analyzing the structure of the function field one could along the line of the present paper establish maximum non-$i$-qualifying sets. Strong candidates for such investigations include codes defined from the Hermitian curve and in larger generality codes defined from norm-trace curves. Among other families of codes that may support similar studies we mention nested more-point algebraic geometric codes, nested codes from the projective counter part to the monomial-Cartesian code construction, nested BCH-codes as well as nested matrix-product codes.

In a related research to the material in Section 2 and Appendix~\ref{secproofof} we obtained the forthcoming paper~\cite{geilmanuscript2026} solving a long standing problem of unifying theory for common affine roots of multivariate polynomials over any field.
\appendix

\section{Proof of Theorem~\ref{thegen}}\label{secproofof}

\begin{proof}
For the reconstruction numbers we observe that for any set $A$ with $\# A > n-M_j(C_1,C_2)$ the corresponding value of~(\ref{eqmaxD}) is at most equal to $j-1$ and that some set $A$ of size $n-M_j(C_1,C_2)$ it holds that~(\ref{eqmaxD}) equals $j$. Therefore
$$r_{\ell -(j-1)}=r_{\ell-j+1}=n-M_j(C_1,C_2)+1.$$
Substituting $j$ with $\ell -m+1$ we obtain the desired result.\\
For the privacy numbers we are interested in the minimal cardinality of a set $A$ such that 
\begin{equation}
m=\ell-(\dim(C_1)_{\bar{A}}-\dim(C_2)_{\bar{A}}),\label{eqreff}
\end{equation}
giving which we will be able to conclude $t_m=\# A-1$. By~(\ref{eqfn}) we have
$$\ell=\dim(C_1)_{\bar{A}}+\dim P_{A}(C_1)-\dim (C_2)_{\bar{A}}-\dim P_A(C_2)$$
and combining this with Forney's second duality lemma~\cite{forney94}[Lem.\ 2] which reads
$$\dim P_A(C)+\dim (C^\perp)_A=\# A$$
the right hand side of~(\ref{eqreff}) becomes
\begin{eqnarray}
&&\dim P_{{A}}(C_1)-\dim P_{{A}} (C_2) \nonumber \\
&=&(\# A-\dim (C_1^\perp)_{{A}})-(\# A -\dim (C_2^\perp)_{{A}}) \nonumber \\
&=&\dim (C_2^\perp)_{{A}}-\dim (C_1^\perp)_{{A}}\nonumber \\
&=&\max \{ \dim D \mid D \subseteq C_2^\perp, D \cap C_1^\perp =\{ \vec{0}\}, {\mbox{Supp}}(D) \subseteq A\} \label{eq27}.
\end{eqnarray}
Obviously, for $\#A < M_m(C_2^\perp, C_1^\perp)$ the value in~(\ref{eq27}) is strictly less than $m$, and for some $A$ of size $M_m(C_2^\perp, C_1^\perp)$ equality holds. This concludes the proof.
\end{proof}

\bibliographystyle{plain}


%
%


%
%

\end{document}